\documentclass[12pt,a4paper]{article}
\usepackage[dvipdfmx,hiresbb]{graphicx}
\usepackage[dvipdfmx]{color}

\usepackage{color}
\usepackage{bm,enumerate,amsmath,amssymb,amsthm}
\usepackage{epsfig}
\usepackage{cite}
\usepackage{algorithm}
\usepackage{algpseudocode}
\usepackage{txfonts}
\usepackage{subfigmat}
\usepackage{figsize}
\usepackage{here}
\usepackage{listings}
\usepackage{braket}
\usepackage{comment}
\usepackage{lscape}
\usepackage{bigints}

\usepackage{arxiv}
\usepackage[utf8]{inputenc} 

\theoremstyle{definition}
\newtheorem{theorem}{Theorem}
\newtheorem*{theorem*}{Theorem}
\newtheorem{definition}{Definition}
\newtheorem*{definition*}{Definition}

\newtheorem{lemma}{Lemma}

\newtheorem{remark}{Remark}

\renewcommand{\headeright}{}

\title{Rademacher Complexity Bounds for Parameterized Quantum Circuits Generated by Pauli Strings}

\author{
  Hiroshi Ohno\\
  Toyota Central R \& D Labs., Inc.\\
  Aichi, Japan\\
  \texttt{oono-h@mosk.tytlabs.co.jp}\\
}

\date{\empty}

\begin{document}

\maketitle

\begin{abstract}
  In this study, we analyze the Rademacher complexity $ \mathcal{R}_{M} $ of a parameterized unitary whose generators are chosen from $ n $-qubit Pauli strings.
  Although generalization bounds for quantum machine learning models have been studied in several settings, explicit Rademacher-complexity bounds for parameterized unitaries generated by Pauli strings remain less transparent.
  We derive simple scaling bounds in terms of the number of parameters $ L $ and the number of training samples $ M $: $ \mathcal{O}(\frac{L^{\frac{3}{2}}}{\sqrt{M}}) $ for the full parameter domain and $ \mathcal{O}(\frac{L}{\sqrt{M}}) $ for a restricted parameter domain.
  Furthermore, we compare the obtained results with those for a classical linear model class and suggest a potential statistical-complexity advantage when the norms of both the input and the parameter in the classical model scale with the number of parameters.
  Numerical experiments provide qualitative evidence consistent with the predicted scaling.
\end{abstract}

\section{Introduction}
Quantum machine learning (QML) is an active and challenging research area \cite{schuld2015,biamonte2017,benedetti2019}.
QML models have empirically demonstrated a quantum advantage over classical models \cite{huang2021,liu2021}.
Furthermore, generalization performance and trainability of a parameterized quantum circuit (unitary) generated by a set of generators appear to be theoretically characterizable using Lie algebras \cite{cerezo2021,ragone2024}.

We address the generalization properties of parameterized quantum circuits (unitaries) generated by Pauli strings as QML models, and present an upper bound on the generalization gap using Rademacher complexity.
In addition, regarding generalization, the theoretical possibility of a quantum advantage over classical models appears not to have been fully clarified, which motivates us to investigate this question.

Previous studies \cite{huang2021,caro2022,bu2022} have demonstrated the generalization performance of parameterized quantum circuits in terms of the number of trainable gates, circuit resources, and the number of training samples for QML models.
In this study, we restrict our attention to QML models generated by $ n $-qubit Pauli strings, where the observable used in their expectation-value functions is also a Pauli string.

Our contributions are summarized as follows:
\begin{itemize}
\item The Lipschitz constant is $ \mathcal{O}(\sqrt{L}) $, where $ L $ is the number of parameters in the quantum circuits.
\item Using the covering number and Dudley's entropy integral, we bound the empirical Rademacher complexity by $ \mathcal{O}(\frac{L^{\frac{3}{2}}}{\sqrt{M}}) $, where $ M $ is the number of samples.
\item In a restricted parameter space, the complexity improves to $ \mathcal{O}(\frac{L}{\sqrt{M}}) $.
\item In the numerical experiments, we present a random-search algorithm that gives a lower approximation for the empirical Rademacher complexity, and empirically investigate the dependence on $ M $ and $ L $ using this approximation algorithm.
\end{itemize}

The remainder of this paper is organized as follows.
In Section \ref{sec2}, we describe the relationship to relevant studies.
Section \ref{sec3} presents a useful lemma and derives the empirical Rademacher complexity of parameterized unitaries.
Section \ref{sec4} presents a random-search algorithm for the empirical Rademacher complexity and reports numerical results.
Section \ref{sec5} summarizes this study and discusses future work.

\section{Related work}\label{sec2}
Parameterized quantum circuits constitute a standard class of models, and it is important to study their expressibility, trainability, and generalization performance \cite{benedetti2019}.

Caro et al. \cite{caro2022} showed that the generalization error for QML is upper-bounded in terms of the number of trainable gates $ T $ and the number of training data points (samples) $ N $, and scales as $ \sqrt{\frac{T}{N}} $ in the worst case.
In addition, when $ K \ll T $ gates undergo substantial changes during the optimization process, the bound improves to $ \sqrt{\frac{K}{N}} $.
In contrast, the present study explicitly analyzes the Rademacher complexity for parameterized unitaries generated by Pauli strings, rather than using only the number of gates.

In \cite{bu2022}, the authors applied Rademacher complexity to quantum circuit models and investigated the dependence of the statistical complexity on the resources, depth, width, and the numbers of input and output registers of a quantum circuit.
Their paper evaluated the statistical complexity of quantum circuits with respect to circuit resources.
In our study, we focus on the expectation function class of parameterized quantum circuits restricted to $ n $-qubit Pauli strings and derive an upper bound on the Rademacher complexity with respect to the number of parameters $ L $ and the number of samples $ M $, based on the Lipschitz constant and covering number.

Huang et al. \cite{huang2021} developed a methodology for assessing potential quantum advantage in learning tasks.
They proposed a projected quantum model that provides a quantum speed-up for a learning problem and demonstrated a significant prediction advantage over some classical models on engineered data sets designed to exhibit a maximal quantum advantage.
This line of work demonstrates that quantum advantage in QML should be rigorously assessed using statistical learning theory.

Gil-Fuster et al. \cite{gilfuster2024} showed through numerical experiments that QML models can fit random labels or random training states, indicating that there are limitations to explaining generalization performance based only on the model class.
Our study gives a theoretical generalization-gap bound based only on model-class complexity; therefore, we leave a detailed understanding of practical generalization behavior, including the effects of data distributions, parameter initialization, noise in the data, and related factors, to future work.

\section{Method}\label{sec3}
The complexity of the training model affects its generalization performance.
In general, a lower-complexity model tends to generalize better.

The empirical Rademacher complexity is a well-known complexity measure in machine learning \cite{mohri2018}.
We present its definition as follows:
\begin{definition} [Empirical Rademacher complexity \cite{mohri2018}]\label{def1}
  \begin{equation}
    \mathcal{R}_{M} (\mathcal{F}) \coloneqq \mathbb{E}_{\sigma} \left[ \sup_{\theta \in \Theta} \frac{1}{M} \sum_{j=1}^{M} \sigma_{j} f(x_{j} \, ; \theta) \right],
  \end{equation}
  where the function class $ \mathcal{F} \coloneqq \{ f_{\theta} : x \rightarrow f(x \, ; \theta) \, | \, \theta \in \Theta \} $, $ x_{j} $ denotes an input data point, $ \Theta \subseteq \mathbb{R}^{d} $ denotes the parameter set for $ f $, $ d $ is the number of parameters of $ f $, $ M $ denotes the number of samples, and the $ \sigma_{j} \in \{-1, 1\} $ are independent, uniformly distributed random variables.
\end{definition}

To derive the empirical Rademacher complexity of a parameterized unitary, we use the following useful lemma.
\begin{lemma}\label{lem1}
  Let the parameterized unitary $ U $ be $ \prod_{j=1}^{L} \exp{ (-i \theta_{j} G_{j}) }$, where each generator $ G_{j} $ and the observable $ O $ are $ n $-qubit Pauli strings, and $ L $ is the number of parameters.
  For the expectation value $ f(x \, ; \theta) = \braket{\psi(x) | U^{\dagger} O U | \psi(x)} $ where $ \ket{\psi(x)} $ denotes an input state, a Lipschitz constant $ K $ for $ f $ can be chosen as
  \begin{equation*}
    K \leq 2 \sqrt{L}.
  \end{equation*}
\end{lemma}
\begin{proof}
  For each $ j $, differentiating $ U(\theta) $ with respect to $ \theta_{j} $ introduces an effective Hermitian generator $ \widetilde G_{j}(\theta) $, which is unitarily equivalent to $ G_{j} $.
  Hence $ \| \widetilde G_{j}(\theta) \|_{\rm op} = \| G_{j} \|_{\rm op} = 1 $.
  Then,
  \begin{equation}
    \frac{\partial f}{\partial \theta_{j}} = i \braket{ \psi(x)|U^\dagger [\widetilde G_{j}(\theta),O] U | \psi(x)}.
  \end{equation}
  Therefore,
  \begin{equation}
    \left| \frac{\partial f}{\partial \theta_{j}} \right| \leq \| [\widetilde G_{j}(\theta), O] \|_{\rm op} \leq 2 \| \widetilde G_{j}(\theta) \|_{\rm op} \| O \|_{\rm op} \leq 2.
  \end{equation}
  Thus,
  \begin{equation}
    \| \nabla_\theta f(x \, ; \theta) \|_{2} = \left( \sum_{j=1}^{L} \left | \frac{\partial f}{\partial \theta_{j}} \right|^{2} \right)^{\frac{1}{2}} \leq 2 \sqrt{L}.
  \end{equation}
  Hence, since $ K \coloneqq \sup_{\theta \in \Theta} \| \nabla_{\theta} f(x \, ; \theta) \|_{2} $, $ K \leq 2 \sqrt{L} $.
\end{proof}
Therefore, $ K = 2 \sqrt{L} $ is a valid choice.

The lemma yields the following theorem.
\begin{theorem}\label{th1}
  Assume that each generator $ G_{j} $ and the observable $ O $ are $ n $-qubit Pauli strings, and that $ \theta \in [0, 2\pi]^{L} $.
  For the expectation $ f(x \, ; \theta) = \braket{\psi(x) | U^{\dagger} O U | \psi(x)} $ of the parameterized unitary $ U(\theta) $, the empirical Rademacher complexity is bounded by
  \begin{equation}
    \mathcal{R}_{M} (\mathcal{F}) = \mathcal{O}\left( \frac{L^{\frac{3}{2}}}{\sqrt{M}} \right).
  \end{equation}
\end{theorem}
\begin{proof}
  Using the covering number $ N $, the Rademacher complexity $ \mathcal{R}_{M} (\mathcal{F}) $ is bounded by
  \begin{equation}
    \mathcal{R}_{M} (\mathcal{F}) \leq 4 \alpha + \frac{12}{\sqrt{M}} \int_{\alpha}^{{\rm diam}(\mathcal{F}, d_{S})} \sqrt{ \log N(\epsilon, \mathcal{F}, d_{S}) } \, d\epsilon,
  \end{equation}
  where $ d_{S} \coloneqq \left( \frac{1}{M} \sum_{j = 1}^{M} (f(x_{j} \, ; \, \theta) - f(x_{j} \, ; \, \theta^{\prime}))^{2} \right)^{\frac{1}{2}} $, $ \epsilon > 0 $, and $ \alpha > 0 $ \cite{caro2022,shalev2014}.
  Using a Lipschitz constant $ K $, $ d_{S} \leq K \| \theta - \theta^{\prime} \|_{2} $.
  Then, $ N(\epsilon, \mathcal{F}, d_{S}) \leq N(\frac{\epsilon}{K}, \Theta, \|\cdot\|) $.
  Thus, we obtain
  \begin{equation}
    \mathcal{R}_{M} (\mathcal{F}) \leq 4 \alpha + \frac{12}{\sqrt{M}} \int_{\alpha}^{{\rm diam}(\mathcal{F}, d_{S})} \sqrt{ \log N(\frac{\epsilon}{K}, \Theta, \|\cdot\|) } \, d\epsilon.
  \end{equation}
  Since $ \theta \in [0, 2 \pi]^{L} $, $ {\rm diam}(\mathcal{F}, d_{S}) \leq K {\rm diam}(\Theta) = K \cdot 2 \pi \sqrt{L} = K R $.
  Here, let $ R \coloneqq 2 \pi \sqrt{L} $ denote a diameter of $ [0, 2 \pi]^{L} $ (the parameter space).
  Therefore, $ N(\epsilon, \Theta, \|\cdot\|) \leq \left( \frac{3R}{\epsilon} \right)^{L} $ ($ 0 \leq \epsilon \leq R $).
  Then,
  \begin{equation}
    \mathcal{R}_{M} (\mathcal{F}) \leq 4 \alpha + \frac{12 \sqrt{L}}{\sqrt{M}} \int_{\alpha}^{KR} \sqrt{ \log \frac{3 KR}{\epsilon} } \, d\epsilon.
  \end{equation}
  Using $ \int_{\alpha}^{KR} \sqrt{ \log \frac{3 KR}{\epsilon} } \, d\epsilon \leq \frac{3 \sqrt{\pi}}{2} KR $, we obtain
  \begin{equation}
    \mathcal{R}_{M} (\mathcal{F}) \leq 4 \alpha + 18 \sqrt{\pi}\frac{\sqrt{L}}{\sqrt{M}} KR.
  \end{equation}
  By Lemma \ref{lem1}, we may take $ K = 2 \sqrt{L} $,
  \begin{equation}
    \begin{split}
      \mathcal{R}_{M} (\mathcal{F}) &\leq 4 \alpha + 36 \sqrt{\pi}\frac{LR}{\sqrt{M}}\\
      &= 4 \alpha + 72 \pi^{\frac{3}{2}} \frac{L^{\frac{3}{2}}}{\sqrt{M}}.
    \end{split}
  \end{equation}
  Since the above bound holds for every $ \alpha > 0 $, taking $ \alpha \downarrow 0 $, we obtain
  \begin{equation}
    \mathcal{R}_{M} (\mathcal{F}) \lesssim 72 \pi^{\frac{3}{2}} \frac{L^{\frac{3}{2}}}{\sqrt{M}}.
  \end{equation}
\end{proof}
We note that since $ f(x \, ; \theta) $ when the observable is a Pauli string is bounded in $ [-1, 1] $, the empirical Rademacher complexity is bounded above by one.
Therefore, $ \mathcal{R}_{M} (\mathcal{F}) \leq  \min \left\{1, C \frac{L^{\frac{3}{2}}}{\sqrt{M}} \right\} $, where $ C $ is a positive constant.

\begin{remark}
  If the parameter domain is restricted so that $ R = \mathcal{O}(1) $\footnote{
    For example, if $ \theta \in [0, 2 \pi/\sqrt{L}]^{L} $, then the parameter domain's diameter is $ 2 \pi $, i.e., $ \mathcal{O}(1) $.
  }, the same argument gives $ \mathcal{R}_{M} (\mathcal{F}) = \mathcal{O}\left( \frac{L}{\sqrt{M}} \right) $.
\end{remark}
\begin{remark}
  For the linear model class $ \mathcal{F} \coloneqq \{ x \rightarrow w^{\mathsf{T}} x \; | \; \| x \|_{2} \leq D, \, \| w \|_{2} \leq B \} $, the Rademacher complexity $ \mathcal{R}_{M} (\mathcal{F}) $ is bounded by $ \mathcal{O}( \frac{DB}{\sqrt{M}} ) $.
  Let the size of $ w $ (the number of parameters) be $ p $.
  With respect to $ p $, if both $ B $ and $ D $ are $ \mathcal{O}(1) $, then $ \mathcal{R}_{M} (\mathcal{F}) = \mathcal{O}( \frac{1}{\sqrt{M}} ) $.
  If either $ B $ or $ D $ is $ \mathcal{O}(p) $, then $ \mathcal{R}_{M} (\mathcal{F}) = \mathcal{O}( \frac{p}{\sqrt{M}} ) $.
  If both $ B $ and $ D $ are $ \mathcal{O}(p) $, then $ \mathcal{R}_{M} (\mathcal{F}) = \mathcal{O}( \frac{p^{2}}{\sqrt{M}} ) $.
  Thus, when $ B = D = \mathcal{O}(p) $, the generalization bound for the linear model class may be worse than that of quantum circuits based on Pauli strings, assuming that $ p $ in the linear model corresponds to $ L $ in the quantum circuit.
\end{remark}

\section{Numerical experiments and results}\label{sec4}
In this section, we perform the qualitative consistency check of Theorem~\ref{th1} from the following three perspectives:
\begin{itemize}
\item Generalization gap,
\item Sample-size ($ M $) dependence,
\item Parameter-size ($ L $) dependence.
\end{itemize}
Here, the generalization gap is defined as the difference between the test loss $ L(f) $ and the empirical loss $ \hat{L}_{M}(f) $.
According to \cite{mohri2018}, for any $ \delta > 0 $, the generalization gap is bounded with probability at least $ 1 - \delta $ as follows:
\begin{equation}\label{eq42}
  L(f) - \hat{L}_{M}(f) \leq 2 \mathcal{R}_{M}(\mathcal{F}) + \sqrt{\frac{ \log{\frac{1}{\delta}} }{2M}}.
\end{equation}
Thus, the generalization gap is bounded by the Rademacher complexity.
To investigate the dependence on the sample size and parameter size, we use the definition of the empirical Rademacher complexity given in Definition~\ref{def1}.

Since the supremum over $ \Theta $ is computationally demanding, we approximate it by maximizing over a finite random subset $ \Theta_{rand} \subset \Theta $ as follows:
\begin{equation}\label{eq41}
  \mathcal{R}_{M} (\mathcal{F}) \approx \mathbb{E}_{S \sim \mathcal{D}^{M}} \mathbb{E}_{\sigma} \left[ \max_{\theta \in \Theta} \frac{1}{M} \sum_{j=1}^{M} \sigma_{j} f(x_{j} \, ; \theta) \right],
\end{equation}
where $ S = \{x_{j}\}_{j=1}^{M} $ and $ \mathcal{D}^{M} $ denotes an $ M $-dimensional data distribution.
The resulting value is a lower approximation of the dataset-averaged empirical Rademacher complexity.
In the numerical experiments, we resample $ S $ at each trial and report the average over trials.
Therefore, the reported quantity should be interpreted not as the empirical Rademacher complexity for a single fixed dataset, but as a Monte Carlo estimate of its expectation over randomly drawn datasets.
This procedure reduces dependence on a particular realization of the input sample.

To compute the generalization gap, we generated input--output training pairs as follows.
The output state of a quantum circuit, which served as a target model, was obtained as $ \ket{\phi} = \left( \prod_{j=1}^{L} \exp(-i \theta_{j} G_{j}) \right) \, RY(x)^{\otimes n} \, \ket{0}^{\otimes n} $.
The corresponding target value was obtained as $ \hat{y} = \braket{\phi | \, Z(0) \, | \phi} $.
Here, the number of qubits $ n $ was set to five, $ L $ was set to five, the input value $ x $ was uniformly sampled from $ [0, 2\pi] $, and $ \theta_{j} $ was uniformly sampled from $ [-0.01, 0.01] $.
The generator $ G_{j} $ was sampled from the set of 5-qubit Pauli strings, excluding the all-identity string.
We repeated this sampling and evaluation procedure $ M $ times to obtain a training dataset of size $ M $.

For the experiments on sample-size and parameter-size dependence, Algorithm~\ref{alg1} gives a random-search procedure for computing a lower approximation of the empirical Rademacher complexity\footnote{
  Quantum circuits were implemented in PennyLane \cite{bergholm2018}.
  Noiseless and state-vector simulations were performed.
}.
Here, in this algorithm, we included the sampling of $ x $, which aims to reduce the selection bias of $ x $.
However, the volume of the parameter space grows exponentially with $ L $.
Thus, Algorithm~\ref{alg1} may be effective only for relatively small $ L $.
\begin{algorithm}
\caption{Random-search approximation used in the numerical experiments}\label{alg1}
\begin{algorithmic}[1]
\Require $ N_{trials} $, $ N_{data} $, $ N_{\theta} $, $ L $, $ {\rm Scale} $, and $ f(x \, ; \theta) $.
\State $ G \leftarrow \emptyset $.
\For{$ t = 1, \ldots, N_{trials} $}
\State $ x_{j} \sim \mathrm{Uniform}(0, 2 \pi) $ for $ 1 \leq j \leq N_{data} $.
\State $ \sigma_{j} \sim \{-1, 1\} $ with probability 0.5 for each value, for $ 1 \leq j \leq N_{data} $.
\State $ S \leftarrow \emptyset $.
\For{$ k = 1, \ldots, N_{\theta} $}
\State $ \theta_{j} \sim \mathrm{Uniform}(0, 2 \pi \cdot {\rm Scale}) $ for $ 1 \leq j \leq L $.
\State Compute $ s = \frac{1}{N_{data}} \sum_{j=1}^{N_{data}} \sigma_{j} f(x_{j} \, ; \theta) $.
\State $ S \leftarrow S \cup \{ s \} $.
\EndFor
\State Find the maximum value in $ S $ and add it to $ G $.
\EndFor
\State Compute the average $ g_{ave} $ and standard error $ g_{err} $ of the elements in $ G $.
\end{algorithmic}
\end{algorithm}

Figure~\ref{fig4-1} shows the generalization-gap results obtained using Eq.~\ref{eq41}.
For training, the number of epochs was set to 5000, and training was terminated when the root mean square error (RMSE) became smaller than $ 1 \times 10^{-4} $.
The optimizer was the simultaneous perturbation stochastic approximation (SPSA) \cite{gacon2021}.
The learning rate was 0.01, and the momentum term was 0.5.
The initial value of $ \theta $ was uniformly sampled from $ [-1, 1] $.
The test sample size was 1000.
The generator was also sampled from the set of 5-qubit Pauli strings, excluding the all-identity string.

For $ L \in \{4, 5, 8\} $, the number of samples $ M $ was set to 25 and 100.
Training was performed 20 times with different random seeds.
\begin{figure}[htbp]
  \centering
  \begin{tabular}{c}
    \includegraphics[width=10cm]{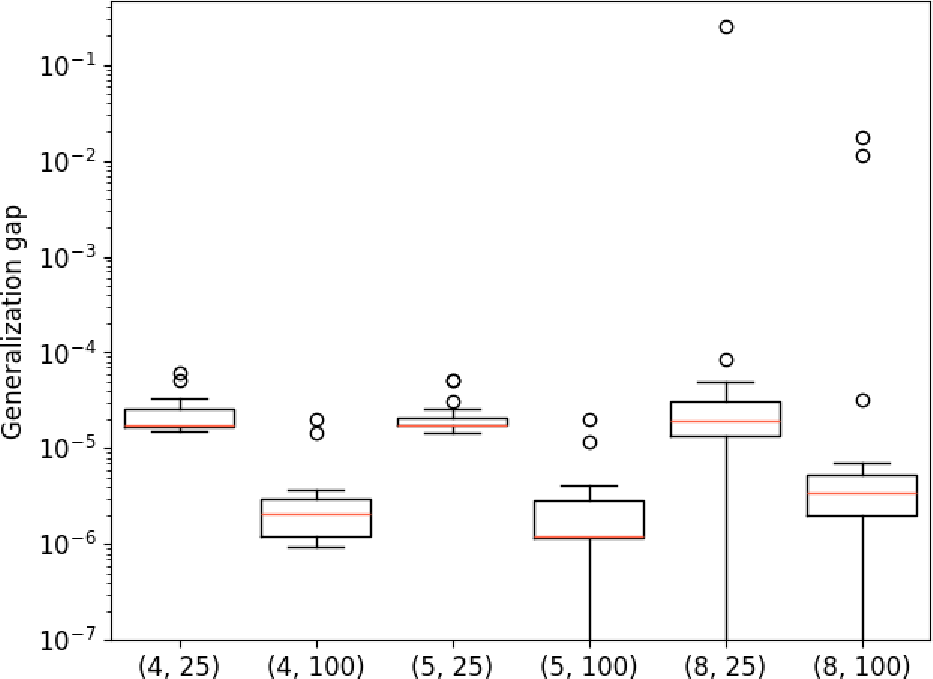}
  \end{tabular}
  \caption{Box plots of the generalization gaps for pairs $ (L, M) $ of parameter size $ L $ and sample size $ M $}\label{fig4-1}
\end{figure}
For each value of $ L $, larger sample sizes reduced the generalization gap.
Regardless of the value of $ L $, the median gaps were approximately the same.
For the larger value of $ L $ ($ L = 8 $), the variance became somewhat larger.
This may be because the training data size $ M $ was relatively small compared with the parameter size $ L $.
In addition, the gaps for $ M = 100 $ were not reduced by a factor of two relative to those for $ M = 25 $, which differs from the theoretical outcome suggested by Eq.~\ref{eq42} and Theorem~\ref{th1}.
This discrepancy may be due to the convergence behavior during training, which depends on the optimizer.

Next, we report the mean values and standard errors for the sample-size dependence of the Rademacher complexity in Eq.~\ref{eq41} in Figure~\ref{fig4-2}.
Here, $ L = 5 $ and $ M \in \{10, 100, 200, 400, 800\} $.
In Algorithm~\ref{alg1}, the hyperparameters were set as follows: $ N_{trials} = 500 $, $ N_{data} \in \{10, 100, 200, 400, 800 \}$, $ N_{\theta} = 100000 $, and $ {\rm Scale} = 1.0 $.
\begin{figure}[htbp]
  \centering
  \begin{tabular}{c}
    \includegraphics[width=10cm]{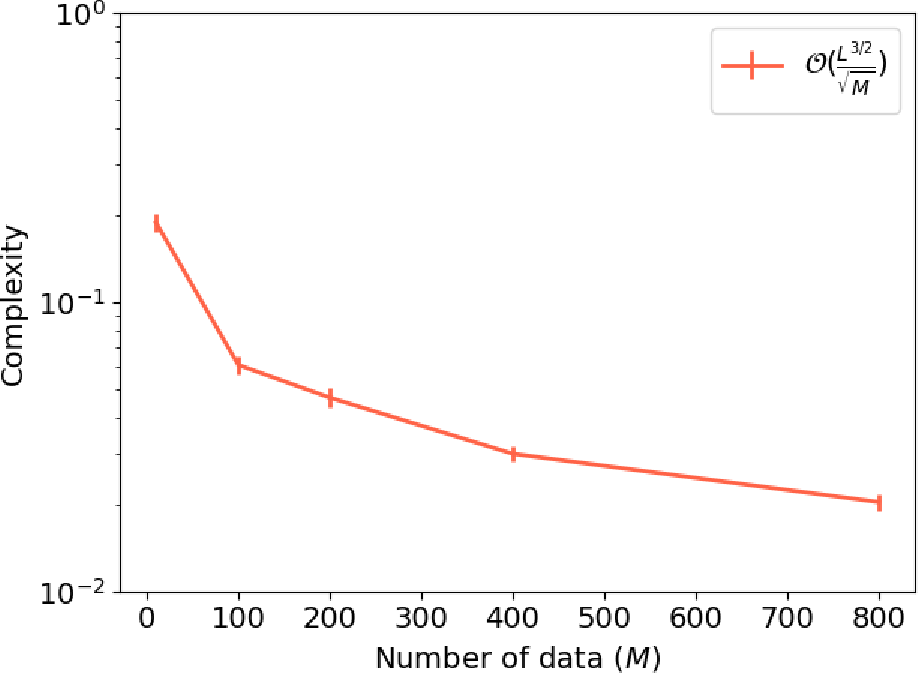}
  \end{tabular}
  \caption{Results for the dependence on sample size $ M $. The y-axis denotes approximated empirical Rademacher complexity.}\label{fig4-2}
\end{figure}
We took logarithms of the approximated Rademacher complexity and $ M $ and performed linear regression analysis.
We obtained a slope of -0.49782 (the p-value $ < 0.01 $).
Therefore, the results are consistent with the $ \frac{1}{\sqrt{M}} $ scaling predicted by Theorem~\ref{th1}.

Finally, Figure~\ref{fig4-3} shows the average values and standard errors for the parameter-size dependence.
In Algorithm~\ref{alg1}, the hyperparameters were set as follows: $ N_{trials} = 500 $, $ N_{data} = 200 $, $ N_{\theta} = 50000L $, and $ {\rm Scale} = 1.0 $ for $ \mathcal{O}(\frac{L^{\frac{3}{2}}}{\sqrt{M}}) $ and $ {\rm Scale} = \frac{1}{2\pi \sqrt{L}} $ for $ \mathcal{O}(\frac{L}{\sqrt{M}}) $.
Here, we assume that $ N_{\theta} $ scales proportionally with $ L $, since increasing $ L $ enlarges the parameter space.
Appendix \ref{ap1} shows the convergence of the approximated complexity obtained by Algorithm~\ref{alg1}.
\begin{figure}[htbp]
  \centering
  \begin{tabular}{c}
    \includegraphics[width=10cm]{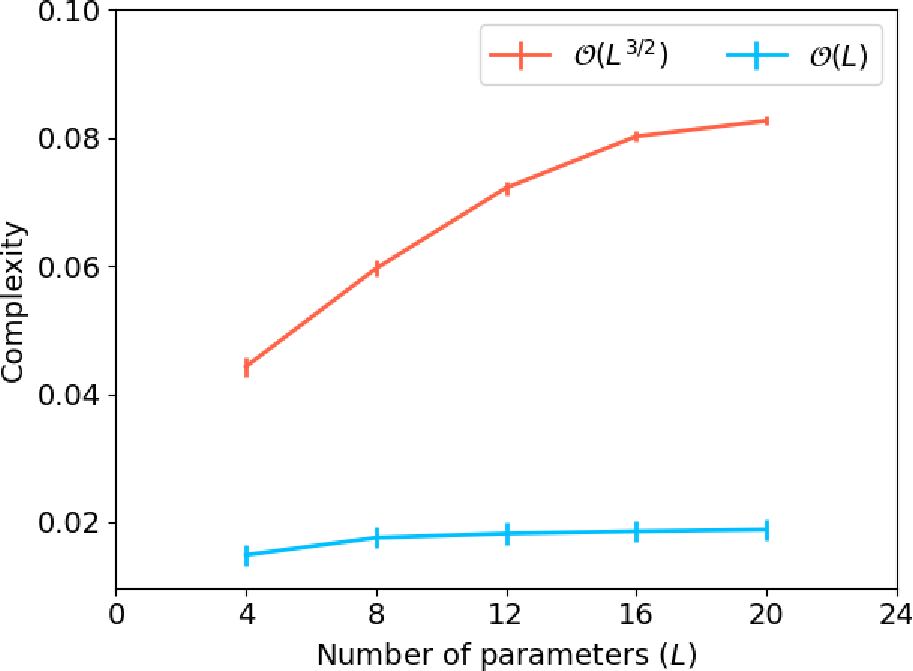}
  \end{tabular}
  \caption{Results for the dependence on parameter size $ L $. The y-axis denotes approximated empirical Rademacher complexity.}\label{fig4-3}
\end{figure}
As shown in the figure, increasing the parameter size generally increased the complexity in both cases.
The complexity for $ \mathcal{O}(\frac{L}{\sqrt{M}}) $ was substantially lower than that for $ \mathcal{O}(\frac{L^{\frac{3}{2}}}{\sqrt{M}}) $, which is qualitatively consistent with the theoretical prediction.
We took logarithms of the approximated Rademacher complexities for $ \mathcal{O}(\frac{L^{\frac{3}{2}}}{\sqrt{M}}) $ and $ \mathcal{O}(\frac{L}{\sqrt{M}}) $ and of $ L $, and performed linear regression analysis for each case.
The slope was 0.40372 (the p-value $ < 0.01 $) for $ \mathcal{O}(\frac{L^{\frac{3}{2}}}{\sqrt{M}}) $ and 0.14498 (the p-value = 0.02547) for $ \mathcal{O}(\frac{L}{\sqrt{M}}) $.
Each value was lower than 1.5 or 1, respectively.
These results are also broadly consistent with the theoretically predicted trend based on $ L^{\frac{3}{2}} $ and $ L $, but do not quantitatively recover the theoretical exponents.
In addition, the values of $ \mathcal{O}(\frac{L^{\frac{3}{2}}}{\sqrt{M}}) $ and $ \mathcal{O}(\frac{L}{\sqrt{M}}) $ exhibited gradually diminishing growth rates (Appendix \ref{ap1}).
This may be due to the insufficient number of $ N_{\theta} $ as the parameter space increases.

\section{Conclusion}\label{sec5}
We analyzed the Rademacher complexity $ \mathcal{R}_{M} $ of a parameterized unitary whose generators are chosen from $ n $-qubit Pauli strings.
The Rademacher complexity of a parameterized unitary generated by Pauli strings is scaled as $ L^{\frac{3}{2}} $; furthermore, when the parameter space is restricted, the order improves to $ L $.
In addition, our results suggest that the resulting complexity has potentially more favorable scaling than a classical linear model class under certain norm-scaling assumptions.
We conducted numerical experiments and obtained results consistent with the theoretical predictions.
However, in Algorithm~\ref{alg1}, the number of $ \theta $ samples ($ N_{\theta} $) may be somewhat small for computing a lower approximation of the empirical Rademacher complexity.
In future work, we will compute the complexity by adopting an optimization algorithm to maximize over $ \theta $ in Eq. \ref{eq41}, rather than relying on the random-search procedure.

\appendix
\section{Convergence of approximated complexity}\label{ap1}
As an example, Figure \ref{figa-1} shows the approximated complexity obtained by Algorithm~\ref{alg1} with varying $ N_{trials} $ for $ L \in \{4, 8, 12, 16, 20\} $, $ N_{data} = 200 $, $ N_{\theta} = 50000 L $, and $ {\rm Scale} = 1.0 $.
\begin{figure}[htbp]
  \centering
  \begin{tabular}{c}
    \includegraphics[width=10cm]{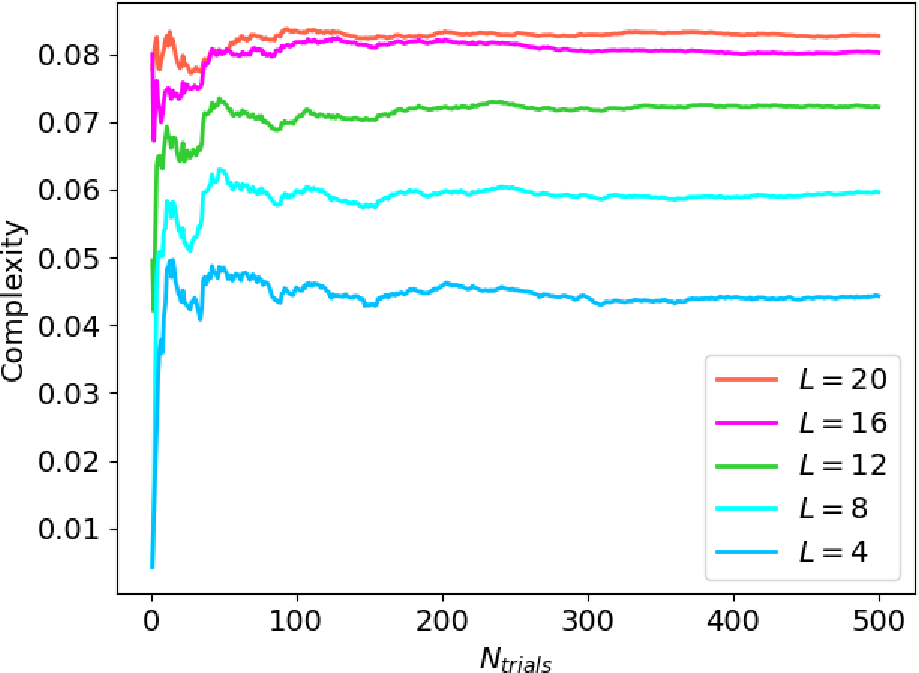}
  \end{tabular}
  \caption{Results for the approximated complexity as a function of $ N_{trials} $}\label{figa-1}
\end{figure}
The figure shows that, with respect to $ N_{trials} $, the approximated complexities converge to flat behavior for each $ L $.
As $ L $ increases, the intervals of converged complexities become narrower.
This may indicate that $ N_{\theta} $ is insufficient relative to the size of the parameter space.

\bibliographystyle{plain}
\bibliography{references}

\end{document}